\documentclass{article}

\usepackage[ruled]{algorithm2e}
\usepackage{amsmath,amsthm,amssymb}
\usepackage{tikz,xcolor, enumitem}
\newtheorem{theorem}{Theorem} 
\newtheorem{definition}{Definition} 
 
\newtheorem{lemma}{Lemma} 
\newtheorem{remark}{Remark} 
\usepackage{balance} 

\begin{document} 

\title{Mechanism Design With Predictions for Obnoxious Facility Location}


\author{Gabriel Istrate\thanks{University of Bucharest, email: gabrielistrate@acm.org}, 
Cosmin Bonchi\c{s}
\thanks{West University of Timi\c{s}oara}}



\maketitle 

\begin{abstract}
We study \emph{mechanism design with predictions} for the obnoxious facility location problem. We present deterministic strategyproof mechanisms that display tradeoffs between robustness and consistency on segments, squares, circles and trees. All these mechanisms are actually group strategyproof, with the exception of the case of squares, where manipulations from coalitions of two agents exist. We prove that these tradeoffs are optimal in the 1-dimensional case. 
\end{abstract}







\section{Introduction}

The theory of \emph{algorithms with predictions} \cite{algo-predictions,purohit2018improving, algo-predictions-webpage} is, without a doubt, one of the most exciting recent research directions in algorithmics: when supplemented by a (correct) \emph{predictor}, often based on machine learning, the newly-developed algorithms are capable of outcompeting their worst-case classical counterparts. A desirable feature of such algorithms is, of course, to perform comparably to the (worst-case) algorithms when the predictors are really bad. This requirement often results \cite{purohit2018improving} in tradeoffs between two measures of algorithm performance, 
\emph{robustness} and \emph{consistency}. A significant amount of subsequent research has followed, summarized by the algorithms with predictions webpage \cite{algo-predictions-webpage}. 

Recently, the idea of augmenting algorithms by predictions has been adapted to the game-theoretic setting of \emph{mechanism design} \cite{xu2022mechanism,agrawal2022learning,gkatzelis2022improved,balkanski2022strategyproof}: indeed, strategyproof mechanisms often yield solutions that are only approximately optimal \cite{hartline2013mechanism}. On the other hand, if the designer has access to a predictor for the desired outcome it could conceivably take advantage of this information by creating mechanisms that  lead to an improved approximation ratio, compared to their existing (worst-case) counterparts. Tradeoffs between robustness and consistency similar to the ones from \cite{purohit2018improving} apply to this setting as well. This is perhaps best exemplified by the recent results in \cite{agrawal2022learning}: the authors considered the classical problem of \emph{facility location problem} in a setting with predictions. Procaccia and Tennenholtz \cite{procaccia2013approximate} had given a 2-approximate mechanism for this problem in the setting without predictions under the utilitarian objective function, and proved it was optimal. Later on, the two-dimensional case of the problem was shown \cite{meir2019strategyproof} to have (optimal) approximation ratio $\sqrt{2}$. Among many  results, Agrawal et al. \cite{agrawal2022learning} gave a mechanism for the 1-dimensional case that was 2-robust and 1-consistent. In the two dimensional case they gave a mechanism with predictions parameterized by $c\in (0,1)$ that realizes an (optimal) $(\frac{\sqrt{2c^2+2}}{1+c},\frac{\sqrt{2c^2+2}}{1-c})$ tradeoff between robustness and consistency. 

We contribute to this research direction by studying, in the setting of mechanism design with predictions,  
another version of facility location: the \emph{obnoxious facility location} problem \cite{aziz2020facility,mei2018mechanism,ibara2012characterizing}. This is the version of facility location where the facility is, in some sense, undesirable (e.g. a garbage dump) and all agents may strive to be located as far as possible from the chosen location. 

\subsection{Contributions and Outline} 

We give deterministic mechanisms that display an identical tradeoff between consistency and robustness for obnoxious facility location on various topologies: one and two-dimensional hypercubes (Algorithm~\ref{minusone-alg}, Theorem~\ref{upper-1-2}), circles, (Algorithm~\ref{minustwo-alg}, Theorem~\ref{upper-circle}), and trees (Algorithm~\ref{minusthree-alg}, Theorem~\ref{upper-trees}). The particular case of two-dimensional hypercubes is completely new (it had not been investigated previously, even without predictions). However, in this case the mechanism that we give is strategyproof but \textbf{not} group strategyproof, being vulnerable to collusion from as little as two agents. 
We give indication that the displayed tradeoffs may be optimal by proving this claim for one-dimensional hypercubes (Theorem~\ref{thm:lb}). We give similar results in the case of agents with dual preferences \cite{zou2015facility, feigenbaum2015strategyproof}, at the expense of replacing group strategyproofness with a weaker concept previously studied (for obnoxious facility location) in \cite{oomine2017}. 

The following is an outline of the paper: In Section~\ref{sec:prelim} we review relevant notations and definitions. The proofs of some of the upper-bound results follow a common strategy: in Section~\ref{sec:ub} we give a general outline of this strategy, and then present the specific details. In Section~\ref{sec:lb} we give details about the lower bounds. We briefly discuss (Section~\ref{sec:related}) related work and (Section~\ref{sec:conclusions}) further directions and open problems. 

\section{Preliminaries}
\label{sec:prelim}

We will write $[n]$ as a substitute for $\{1,2,\ldots, n\}$. As usual, the expectation of a random variable Z will be denoted by $\mathrm{E}[Z]$. 
 
We will be dealing with the \emph{obnoxious facility location problem}. In this problem $n$ agents are located in a metric space $(X,d)$. The location of each agent ($p_i$ for agent $1\leq i\leq n$) is private to the agent. We will refer to elements of $X^{n}$ as \emph{location profiles}. Let $\mathcal{C}=(p_1,p_2,\ldots, p_n)\in X^{n}$ be the vector of (true) agent locations. The goal is to choose a single facility location $f(\mathcal{C})\in X$, via a \emph{mechanism} $f:X^{n}\rightarrow X$.\footnote{As defined the mechanism is deterministic. To model randomized mechanisms we modify the definition of $f$ to $f:X^{n}\times \Omega \rightarrow X$, where $\Omega$ is a probability space.} Once the location has been chosen, each agent gains utlity $u_{i}(f(\mathcal{C}))=d(f(\mathcal{C}),p_i)$. Our goal is to choose a location $P^{*}$ to maximize \emph{obnoxious social welfare} $ SW(P^{*},\mathcal{C}):= \sum_{i=1}^{n} d(P^{*},p_i)$. We will denote this optimal value by $OPT(\mathcal{C})$ and use notation $y^{*}$ for a point realizing it. Given mechanism $f$ and instance $\mathcal{C}$, we will use notation $F(\mathcal{C})$ as a shorthand for $SW(f(\mathcal{C}),\mathcal{C})$. For randomized mechanisms this changes to $F(\mathcal{C})$ for $\mathrm{E}_{\omega \in \Omega}[SW(f(\mathcal{C},\omega),\mathcal{C})]$. We will measure the performance of a mechanism $f$ by the ratio $c(f):= max_{\mathcal{C}\in X^{n}} (\frac{OPT(\mathcal{C})}{F(\mathcal{C})})$. 
 Agents may be incentivized to misreport their location in order to maximize their profit. Given location profile $P$ and set $S\subset [n]$ , define $P_{-S}$ to be the profile obtained by eliminating elements $p_i$, $i\in S$, from $P$. Mechanism $f$ is called \emph{strategyproof} if for every $1\leq i\leq n$ and $p^{\prime}_{i}\neq p_i$, $u_i(f(p_{-i},p^{\prime}_i))\leq u_i(f(P)).$ In other words, it never pays off for agent $i$ to misreport its location as $p^{\prime}_i$. The mechanism is called \emph{group strategyproof} iff for every $S\subset [n]$ and every location profile $(P_{-S},Q_{S})$, there exists an index $i\in S$ such that $u_i(f(P_{-S},Q_S))\leq u_i(f(P)).$ In other words, agents in $S$ cannot all gain by coordinating their deviations from their true locations. 

Let $Z\subseteq X$ be a set of acceptable \emph{prediction values}. 
A \emph{mechanism with predictions is a function} $f:X^{n}\times Z\rightarrow X$ (or $f:X^{n}\times \Omega \times Z\rightarrow X$, in the randomized case). It employs a \emph{predictor} $\widehat{Y}\in Z$, the predicted location of the obnoxious facility, as its last argument. Denote by $\eta=d(P^{*},\widehat{Y})$ the prediction error. We will use the notation $c(f,\eta)$ to refer to the adaptation of $c(f)$ to settings with prediction $(\mathcal{C},P_{C})$ for which the condition $d(y^{*},P_{\mathcal{C}})\leq \eta$ holds. That is, given $\mathcal{C}\in X^{n}$ and $P_{\mathcal{C}}\in Z$, use shorthand $F(\mathcal{C},P_{\mathcal{C}})=SW(f(\mathcal{C},P_{\mathcal{C}}),\mathcal{C})$ and define 
\begin{align} 
c(f,\eta):= max_{\stackrel{\mathcal{C}\in X^{n},P_{C}\in Z}{d(P_{C},y^*)\leq \eta}} \big(\frac{OPT(\mathcal{C})}{F(\mathcal{C},P_{\mathcal{C}})}\big). 
\end{align}

\begin{definition} Given $\gamma,\beta\geq 1$, a mechanism with predictions $f$ is \emph{$\gamma$-robust} iff $c(f,\eta)\leq \gamma$ for all  $\eta$ and is \emph{$\beta$-consistent} if $c(f,0)\leq \beta$. 
\end{definition} 

Following \cite{schummer2002strategy}, we will view graphs as closed, connected subsets of some Euclidean space $\mathbb{R}^{n}$. They are composed of a finite number of segments called edges. We will use, in particular, \emph{trees}. Distances between points in a tree are not computed in the ambient $\mathbb{R}^n$, but rather along tree paths, by summing up the (euclidian) lengths of segments forming the path. Since there is an unique path $[a, b]$ between any two points $a,b$ in a tree one can naturally define, for  arbitrary $\lambda \in [0,1]$ the point $\lambda a+(1-\lambda)b$ as the unique point $w$ on this path such that $d(a,w)=\lambda d(a,b)$. In particular we will denote by $m_{a,b}$ the point $1/2\cdot a+1/2 \cdot b$, the \emph{midpoint of path $ab$.} Vertices $a,b$ in a tree $T$ form a \emph{diameter of $T$} if the distance $d(a,b)$ is the largest distance between any two vertices in the tree. A vertex is called \emph{peripheral} iff it is part of a diameter of the tree. 

Finally, we will refer in the sequel to \emph{majority voting}, used to choose between two items. We will used a slightly generalized version by allowing nonintegral numbers of votes for the items. 

\section{Results}

As in \cite{purohit2018improving}, our mechanisms will be parameterized by $\lambda \in [0,1]$, which intuitively measures the "agresiveness" of relying on the predictor. Our guarantees for consistency/robustness will be functions of $\lambda$ as well. When $\lambda = 1$ we will get the case of the baseline algorithm which simply returns the predicted point. In general this algorithm is 1-consistent but \textbf{not} $\gamma$-robust for any constant $\gamma$. Indeed, let $P$,$Q$ be points in $X$ that are farthest apart. Choose the $n$ points to be all at the same location $P$ and let the predictor predict also $P\in X$. The optimal obnoxious social welfare would be $n\cdot d(P,Q)$, corresponding to locating the facility at $Q$, but the bad predictor leads to an obnoxious social welfare of 0. 
We will plot the consistency/robustness curves against parameter $\lambda\in [0,1]$. By the previous discussion, at $\lambda=1$ $\beta(1)=1$ and $\gamma(1)=\infty.$ On the other hand at $\lambda = 0$ we expect the quantities $\beta(0)$ and $\gamma(0)$ to coincide, and be equal to the approximation ratio of the optimal strategyproof mechanism. Moreover, as we increase $\lambda$ from $0$ to $1$ we expect the robustness $\gamma(\lambda)$ to diverge, while the consistency $\beta(\lambda)$ to tend towards 1. Our goal is to obtain explicit curves $(\beta(\lambda),\gamma(\lambda))$ that display "the best" tradeoff between robustness and consistency. 

The first setting for our results is the one where $X=[0,1]^{k}$, the $k$-dimensional hypercube. 
We will employ the "robust coordinatewise voting" mechanism displayed in Algorithm~\ref{minusone-alg}. We have: 

\begin{algorithm}
    \SetKwInOut{Input}{Input}
    \SetKwInOut{Output}{Prediction}

    \Input{$C=(x_1,x_2,\ldots, x_n), x_i\in [0,1]^k$, $\lambda\in [0,1)$.}
    \Output{$P_{\mathcal{C}}=p_1p_2\ldots p_k\in \{0,1\}^{k}$. W.l.o.g. $P_{\mathcal C}=1^{k}$.} 
       for i:=1 to k \\
       \hspace{5mm}let $n_{1,k}$ be the number of indices $i$ s.t. $0\leq x_{i,k}\leq 1/2$. \\
       \hspace{5mm}let $n_{2,k}$ be the number of indices $i$ s.t. $1/2< x_{i,k}\leq 1$. \\
       \hspace{5mm}count $n_{1,k}$ votes for 1; count $n_{2,k}$ votes for 0. \\
       \hspace{5mm}\textbf{count $\lambda n$ votes for $p_i$.}\\
       \hspace{5mm}if $p_i$ has at least as many votes as $1-p_i$: \\
       \hspace{10mm}let $z_i=p_i$.  \\
       \hspace{5mm}else \\
       \hspace{10mm}let $z_i=1-p_i$. \\
       return $Z_{\mathcal{C}}:=z_1z_2\ldots z_k.$ \\
    \caption{$k$-dimensional robust coordinatewise voting.}
    \label{minusone-alg}
\end{algorithm}

\begin{figure}
\begin{center} 
\includegraphics[width=5cm,height=4cm]{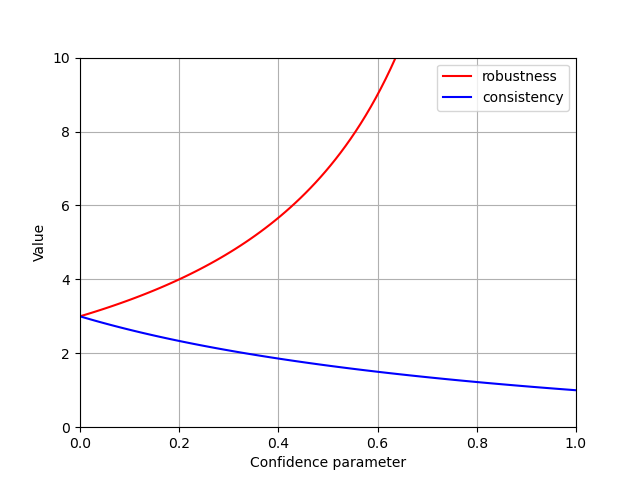}
\begin{tikzpicture}[
roundnode/.style={circle, draw=green!60, fill=green!5, very thick, minimum size=7mm},
squarednode/.style={rectangle, draw=red!60, fill=red!5, very thick, minimum size=5mm},
]
\draw (2,0) -- (0,0);
\draw (0,2) -- (0,0);
\draw (2,0) -- (2,2);
\draw (0,2) -- (2,2); 
\draw (1,0) -- (1,2);
\draw (0,1) -- (2,1);
\draw (0.5,0.5) node{$n_2$};
\draw (1.5,0.5) node{$n_3$};
\draw (1.5,1.5) node{$n_4$};
\draw (0.5,1.5) node{$n_1$};
\filldraw[black] (2,2) circle (2pt) node[anchor=south]{D};
\filldraw[black] (1,1) circle (2pt) node[anchor=west]{R};
\filldraw[black] (1,2) circle (2pt) node[anchor=south]{T};
\filldraw[black] (2,1) circle (2pt) node[anchor=west]{S};
\filldraw[black] (0,1) circle (2pt) node[anchor=east]{U};
\filldraw[black] (1,0) circle (2pt) node[anchor=north]{V};
\filldraw[black] (2,0) circle (2pt) node[anchor=west]{A};
\filldraw[black] (0,0) circle (2pt) node[anchor=east]{B};
\filldraw[black] (0,2) circle (2pt) node[anchor=south]{C};
\end{tikzpicture}
\end{center} 
\caption{ (a). Consistency versus robustness. (b). Decomposition of the square into rectangles used in the case of two-dimensional hypercubes.} 
 \label{plot1}
\end{figure}
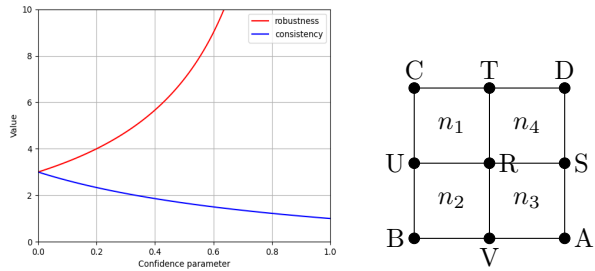 

\begin{theorem} 
The robust majority mechanism in one dimension with parameter $\lambda$ (see Algorithm~\ref{plot1}) is group strategyproof, $\frac{3+\lambda}{1-\lambda}$-robust and $\frac{3-\lambda}{1+\lambda}$-consistent. The robust majority mechanism in \textbf{two} dimensions with parameter $\lambda$ is $\frac{3+\lambda}{1-\lambda}$-robust and $\frac{3-\lambda}{1+\lambda}$-consistent (see Figure~\ref{plot1} (a)). It is strategyproof but \textbf{not} group strategyproof: there exists a coalition of two agents that can manipulate it. 
\label{upper-1-2}
\end{theorem} 

An analogous result holds in the case of circles. We will regard a circle $\mathcal{C}$ as the segment $[0,1]$ where 0 and 1 are identified. The distance between two points with coordinates $0\leq x < y\leq 1$ will be the distance between the points as measured on the circle. That is: $d(x,y)=min(y-x,1+x-y)$. We will employ the mechanism in Algorithm~\ref{minustwo-alg}, and prove that: 

\begin{theorem} 
The robust majority mechanism on a circle with parameter $\lambda$ is group strategyproof, $\frac{3+\lambda}{1+\lambda}$-robust and $\frac{3-\lambda}{1+\lambda}$-consistent. 
\label{upper-circle} 
\end{theorem} 

\begin{algorithm}
    \SetKwInOut{Input}{Input}
    \SetKwInOut{Output}{Prediction}

    \Input{$\mathcal{C}=(x_1,x_2,\ldots, x_n)$, $\lambda\in [0,1]$.}
    \Output{$P_{\mathcal{C}}\in [0,1)$.}
       Let $S=\{x_{1},\ldots, x_{n}\}$. \\
       Let $Q$ be the point diametrally opposed to $P_{\mathcal{C}}$. \\
       Let $n_Q$ be the number of points of $S$ in $(P_{\mathcal{C}}+1/4,P_{\mathcal{C}}+3/4]$ \\ and $n_P$ the number of points of $S$ 
       in $(P_{\mathcal{C}}-1/4,P_{\mathcal{C}}+1/4]$. \\
       If $n_P\leq n_Q+\lambda n$ return point $P_{\mathcal{C}}$\\
       else return point $Q$.  
    \caption{The robust majority mechanism on a circle.}
    \label{minustwo-alg}
\end{algorithm}

\begin{algorithm}
    \SetKwInOut{Input}{Input}
    \SetKwInOut{Output}{Prediction}

    \Input{$x_1,x_2,\ldots, x_n\in T$, $\lambda\in [0,1]$.}
    \Output{$P_{\mathcal{C}}\in T,$ a vertex of $T.$ }
       Let $S=\{x_{1},\ldots, x_{n}\}$. Transform tree $T$ to $T^{\prime}$ by making all points of $S$ vertices, while preserving distances. \\
       Let $a$ be the farthest vertex from $P_{\mathcal{C}}$. 
       Let $b$ be a vertex of $T^{\prime}$ farthest from $a$ ($P_{\mathcal{C}}$ if possible)  \\ 
       Let $T_a$ be the connected component of $T^\prime$ determined by $m_{a,b}$ containing $a$, and $T_b$ be the complement of $T_a$ in $T^{\prime}$. \\
       Let $n_1$ be the number of points of $S$ in $T_a$ and $n_2$ the number of points 
       in $T_b$. \\
       If $n_1+\lambda n \leq n_2$ return point $a$ \\
       else return point $b$.  
    \caption{The robust majority mechanism on a tree.}
    \label{minusthree-alg}
\end{algorithm}

Finally, we propose a version of the robust majority mechanism with parameter $\lambda$ for trees, displayed in Figure~\ref{minusthree-alg}. We have: 

\begin{theorem} 
Under the hypothesis that the predicted point is always a peripherical point, the robust majority mechanism on trees with parameter $\lambda$ is (also) group strategyproof, $\frac{3+\lambda}{1+\lambda}$-robust and $\frac{3-\lambda}{1+\lambda}$-consistent. 
\label{upper-trees} 
\end{theorem} 

\begin{remark} 
Even though the algorithm from Theorem~\ref{upper-trees} displays the same tradeoff as the ones from theorems~\ref{upper-1-2} and~\ref{upper-circle}, it has an important drawback, not present in the other cases: the predicted point $P_{\mathcal{C}}$ is not necessarily among the possible outputs of the algorithm. Indeed (see Lemma~\ref{lemma1} below) in the case of 1-d and 2-d hypercubes the distance is a convex function, hence the optimum was guaranteed \emph{a priori} to be one of the points 0/1 ($A,B,C,D$ in the two-dimensional case). On the other hand in the case of circles, we specifically designed the algorithm so that $P_{\mathcal{C}}$ is one of the potential outputs. To make the algorithm behave similarly to the previous algorithms, we were forced to add the condition that $P_{\mathcal{C}}$ is peripheral to the specification of the predictor: while the robustness bound carries on without this assumption, we don't know how to obtain the consistency bound without it. 

The problem is that there are examples (see e.g. in Exercise 3.4 \cite{abueidacentrality}) where the optimal point $y^{*}$ need \textbf{not} be a peripheral point. In light of this fact, while it is natural to require that the prediction point $P_{\mathcal{C}}$ be a leaf of the tree, since optima are always reached at a leaf \cite{church1978locating},  it seems substantially less natural  to require that $P_{\mathcal{C}}$ is peripheral. One can compute the optimum point in linear time \cite{ting1984linear}. See also \cite{oomine2016characterizing} for a characterization of possible outputs of group strategyproof mechanisms. 
\label{rem1}
\end{remark}

The tradeoff we displayed in Theorems~\ref{upper-1-2} is (at least in the one-dimensional case) optimal. Indeed, we prove: 

\begin{theorem} 
For every $0<c\leq 2$ and $\delta >0$, a deterministic mechanism with predictions for obnoxious facility location on a segment that always yields one of the endpoints and  
is $(1+c)$-consistent cannot be $(1+\frac{4}{c}-\delta)$-robust.  The mechanism in Algorithm~\ref{minusone-alg} realizes the best possible tradeoff, with substitution $\lambda = \frac{2-c}{2+c}$ (i.e. $c=\frac{2(1-\lambda)}{1+\lambda}$). 
\label{thm:lb}
\end{theorem} 
\subsection{The case of dual preferences} 

In \cite{zou2015facility} and \cite{feigenbaum2015strategyproof} yet another twist on the obnoxious facility location was proposed: the case of \emph{dual preferences}. Specifically, agents are of two types: those of type 1 prefer to be as close to the facility as possible, while those of type 0 still want to maximize distance to the facility. Denoting by $T_1, T_0$ the sets of agents of type 1, type 0, respectively, one defines the social welfare of a given choice $Y$ as $SW(P,Y)=\sum_{i\in T_1} (1-d(Y,x_i)) + \sum_{j\in T_0} d(Y,x_j)$. In this case, an agent $i$ may misreport \emph{both} its type (we will denote by $y_i$) and location $x_i$. However, it was noted that the problem is related to the obnoxious facility location as follows: define for each agent $i$ its \emph{transformed location $x^{*}_{i}$} as follows: $x^{*}_{i}=1-x_i$ if $y_i=1$, $x^{*}_{i}=x_i$, otherwise. Zou and Li prove that applying the majority voting mechanism on the transformed location profile yields a group-strategyproof, 3-approximation algorithm.\footnote{The proof does \textbf{not} follow by reduction but needs a direct argument, since in general $1-d(Y,x_i)\neq d(Y,x^{*}_{i})$. However (\cite{zou2015facility}, Lemma 4) $d(Y,x_i)+ d(Y,x^{*}_{i})\leq 1$.} One can similarly extend our Theorem~\ref{minusone-alg}. Interestingly enough, to obtain this we \textbf{don't} extend Algorithm~\ref{minusone-alg}, but need to devise a new mechanism. On the other hand we also have to give up group strategyproofness, and replace it with the following weaker concept: 

\begin{definition} \cite{oomine2017} Given constant $\gamma\geq 1$, a mechanism $f$ is called \emph{$\gamma$-group strategyproof} iff for every $S\subset [n]$ and every location profile $(P_{-S},Q_{S})$, there exists an index $i\in S$ such that $u_i(f(P_{-S},Q_S))\leq \gamma\cdot u_i(f(P),p_i).$ In other words, agents in $S$ cannot all improve their utilities by more that a multiplicative factor of $\gamma$ by coordinating their deviations from their true locations. 
\end{definition} 

\begin{theorem} The transformed robust majority mechanism (Algorithm~\ref{alephzero-alg}) is $\frac{1+\lambda}{1-\lambda}$-group strategyproof, $\frac{3+\lambda}{1-\lambda}$-robust and $\frac{3-\lambda}{1+\lambda}$-consistent. 
\label{upper-dual}
\end{theorem}

\begin{algorithm}
    \SetKwInOut{Input}{Input}
    \SetKwInOut{Output}{Prediction}

    \Input{$C=(x_1,y_1,x_2,y_2,\ldots, x_n,y_n)$, \\ $x_i\in [0,1]$, $y_i\in \{0,1\}$, $\lambda\in [0,1)$.}
    \Output{$P_{\mathcal{C}}=p_1\in \{0,1\}$. w.l.o.g. assume $P_{\mathcal C}=1$.} 
       \hspace{5mm} define \emph{transformed locations} $x^{*}_i=x_i(1-y_i)+(1-x_i)y_i$.\\
       \hspace{5mm}let $n_{1}$ be the number of indices $i$ s.t. $0\leq x^{*}_{i}\leq (1-\lambda)/2$. \\
       \hspace{5mm}let $n_{2}$ be the number of indices $i$ s.t. $(1-\lambda)/2< x^{*}_{i}\leq 1$. \\
       \hspace{5mm}if $n_1> n_2$ return $Z_{\mathcal{C}}:=p_1(=1)$ \\
       \hspace{5mm}else return $Z_{\mathcal{C}}:=1-p_1(=0)$. \\
    \caption{Transformed Robust Majority.}
    \label{alephzero-alg}
\end{algorithm}

\section{Proof Idea for Theorems~\ref{upper-1-2}, ~\ref{upper-circle} and~\ref{upper-trees}}
\label{sec:ub} 

The proofs of Theorems~\ref{upper-1-2},~\ref{upper-circle} and~\ref{upper-trees} proceed by first investigating group strategyproofness directly, and then \textbf{follow a similar pattern to upperbound robustness/consistency:} for an input configuration $\mathcal{C}$ and prediction $P_{\mathcal{C}}$ we want to upper bound the ratio $app(\mathcal{C},P_{\mathcal{C}}):= \frac{OPT(\mathcal{C})}{F(\mathcal{C},P_{\mathcal{C}})}$ by a constant $\gamma(\lambda)$ ($\beta(\lambda)$ for consistency). To do so we will identify a "bad" configuration\footnote{we emphasize $\mathcal{C}_1$ will \textbf{not} necessarily be of the same type as $\mathcal{C}$. That's why we will write $SW(\mathcal{C}_1, f(\mathcal{C},P_{\mathcal{C}}))$ instead of $F(\mathcal{C}_1,P_{\mathcal{C}})$ in inequalities of~(\ref{ub-robustness}).} $\mathcal{C}_1$ and prove that: 
\begin{align} 
\frac{OPT(\mathcal{C})}{F(\mathcal{C},P_{\mathcal{C}})}\leq \frac{SW(\mathcal{C}_1,y^{*})}{SW(\mathcal{C}_1, f(\mathcal{C},P_{\mathcal{C}}))}\leq \frac{OPT(\mathcal{C}_1)}{SW(\mathcal{C}_1, f(\mathcal{C},P_{\mathcal{C}}))}\leq \gamma(\lambda)
\label{ub-robustness}
\end{align} 
where $y^{*}$ is the point realizing the optimum for configuration $\mathcal{C}$
(we replace $\gamma(\lambda)$ by $\beta(\lambda)$ in the inequalities above in consistency proofs). 
The last inequality in~(\ref{ub-robustness}) will typically follow easily by computing the optimal solutions for $\mathcal{C}_1$. The second is clear by the definition of $OPT$. As for the first inequality, since $OPT(\mathcal{C})\geq F(\mathcal{C},P_{\mathcal{C}})$, we will attempt to bound the left-hand side by  
\begin{align*} 
\frac{OPT(\mathcal{C})}{F(\mathcal{C},P_{\mathcal{C}})}\leq \frac{OPT(\mathcal{C})-(F(\mathcal{C},P_{\mathcal{C}})- SW(\mathcal{C}_1, f(\mathcal{C},P_{\mathcal{C}}))}{F(\mathcal{C},P_{\mathcal{C}})-(F(\mathcal{C},P_{\mathcal{C}})- SW(\mathcal{C}_1,f(\mathcal{C},P_{\mathcal{C}}))}
\end{align*}
To make the inequality work we first need to prove that $F(\mathcal{C},P_{\mathcal{C}})- SW(\mathcal{C}_1, f(\mathcal{C},P_{\mathcal{C}}))\geq 0.$ This generally follows automatically by requiring that points of $\mathcal{C}_1$ are translates of the corresponding points of $\mathcal{C}$ towards the chosen location $f(\mathcal{C},P_{\mathcal{C}})$.  
 To make the denominator of the right-hand side less or equal to $SW(\mathcal{C}_1,y^{*})$, given that $OPT(\mathcal{C})=SW(\mathcal{C},y^{*})$,  we will also need that $\Delta:=F(\mathcal{C},P_{\mathcal{C}})-SW(\mathcal{C}_1,f(\mathcal{C},P_{\mathcal{C}}))- SW(\mathcal{C},y^{*})+SW(\mathcal{C}_1,y^{*})\geq 0.$ This inequality will be proved separately in each case. Proofs may employ triangle inequalities for the metric space in question. 

\section{Proof of Theorem~\ref{upper-1-2}}

\subsection{One-dimensional hypercubes (segments)}

\textbf{Robustness:} Consider an arbitrary configuration $\mathcal{C}$ and without loss of generality assume that $P_{\mathcal{C}}=1$ (if this is not the case, simply switch labels 1 and 0). Let $n=(n_1,n_2)$ be the type of $\mathcal{C}$. 
Let $y^{*}$ be the point realizing the optimum social welfare for $\mathcal{C}$. Assume first  that $n_1+ \lambda\cdot n \geq n_2$ (i.e. Algorithm~\ref{minusone-alg} outputs $P_{\mathcal{C}}$). The above condition is equivalent to $n_2/n_1\leq \frac{1+\lambda}{1-\lambda}$ or $n_2/n\leq \frac{2}{1-\lambda}$.  Define $\mathcal{C}_1$ to consist of  $n_1$ points at 0.5 and $n_2$ points at 1. $\mathcal{C}_1$ has the same type as $\mathcal{C}$, and the points in $\mathcal{C}_1$ are at least as close to $P_{\mathcal{C}}$ as their counterparts in $\mathcal{C}$.  The best we could do for $\mathcal{C}_1$ is choose\footnote{0 has social welfare $n_2+n_1/2$, 1 has social welfare $n_1/2$.} $y^{*}=0$, with social welfare $n_2+n_1/2$. The social welfare of the solution of the modified majority mechanism on input $\mathcal{C}_1$ is, on the other hand, $n_{1}/2$. So the approximation ratio will be upperbounded by
\begin{align}
\frac{OPT(\mathcal{C}_1)}{SW(\mathcal{C}_1,f(\mathcal{C},P_{\mathcal{C}}))}= \frac{\frac{n_1}{2}+n_2}{\frac{n_1}{2}}\leq 1+\frac{2(1+\lambda)}{1-\lambda}=\frac{3+\lambda}{1-\lambda} 
\end{align}
\noindent \textbf{Proving $\mathbf{\Delta \geq 0:}$} $\Delta = \sum_{i=1}^{n} |x_i-1| - n_{1}/2 -\sum_{i=1}^{n} |y^{*}-x_i| +  $
\begin{align*} 
 + n_1\cdot |1/2-y^{*}|+ n_{2}\cdot |1-y^{*}| = \sum_{i:x_i\leq 1/2} (1-x_i - 1/2+ |1/2-y^{*}| \\
 -|y^{*}-x_i|) +  \sum_{i: x_i>1/2} (|1-x_i|+ |y^{*}-1|-|y^{*}-x_i|)\geq 0
\end{align*} 
by triangle inequalities $|x_i-1/2|+|1/2-y^{*}|\geq |y^{*}-x_i|$ and similar ones for the second sum.  

In the case $n_1+\lambda n<n_2$ (i.e. Algorithm~\ref{minusone-alg} outputs $0$) take $\mathcal{C}_1$ to consist of $n_1$ points at 0 and $n_2$ points at 1/2. 
\begin{align}
\frac{OPT(\mathcal{C}_1)}{SW(\mathcal{C}_1,f(\mathcal{C},P_{\mathcal{C}}))}= \frac{\frac{n_2}{2}+n_1}{\frac{n_2}{2}}
< 1+ 2\frac{1-\lambda}{1+\lambda}= \frac{3-\lambda}{1+\lambda}\leq \frac{3+\lambda}{1-\lambda}
\end{align} 
\textbf{Proving $\mathbf{\Delta \geq 0:}$} $\Delta = \sum_{i=1}^{n} |x_i-0| - n_{2}/2 -\sum_{i=1}^{n} |y^{*}-x_i| +  $
\begin{align*}
 \sum_{x_i\leq 1/2} |0-y^{*}| + \sum_{x_i>1/2} |1/2-y^{*}| = \sum_{x_i\leq 1/2} (|x_i-0|+  |0-y^{*}|  - \\ - |x_i-y^{*}|)+ \sum_{x_i >1/2} (|x_i-1/2|+ |1/2-y^{*}| - |x_i-y^{*}|)\geq 0 
\end{align*}

\noindent\textbf{Consistency:} Without loss of generality, we assume that the predicted point (which also realizes the optimum) is $P_{\mathcal{C}}=1$. We will also assume that the mechanism~\ref{minusone-alg} yields the opposite of the predicted point, otherwise the consistency ratio is 1.  That is $P_{\mathcal{C}}=y^{*}=1$ but $Z_{\mathcal{C}}=0$. Given our assumptions on the predictor and algorithm, the only possible case is $n_1+\lambda n< n_2$.  We have
\begin{align}
\frac{OPT(\mathcal{C})}{F(\mathcal{C},P_{\mathcal{C}})}= \frac{\sum\limits_{i\leq n} (1-x_i)}{\sum\limits_{i\leq n} x_i}\leq  \frac{2n}{n_2}-1<\frac{4}{1+\lambda}-1= \frac{3-\lambda}{1+\lambda} 
\end{align}
(we have used the inequality $n_2/n > \frac{1+\lambda}{2}$, which follows easily by taking $n_1=n-n_2$, as well as $\sum_i x_i\geq n_2/2$). 
 
\textbf{Group Strategyproofness:} Assume that for the truthful location profile it holds that $n_1+\lambda \cdot n \geq n_2$ (thus Algorithm~\ref{minusone-alg} returns point 1). Let $S\subset [n]$ be a set of agents. If the Algorithm also returns 1 on the modified profile then no agent has improved their utility. Therefore we will assume that agents in $S$ colluding leads to changing the outcome of Algorithm~\ref{minusone-alg} to 0. If $S$ contains an agent in $[0,1/2]$ then the utility of this agent has not improved by colluding. Therefore, the only case left is that all agents in $S$ are in $(1/2,1]$. But then reporting a different location in $(1/2,1]$ does not improve utility, and reporting a location in $[0,1/2]$ may actually decrease it.  The case $n_1+\lambda \cdot n < n_2$ is similar. 

\subsection{Two-dimensional hypercubes (squares)} 
First we note the following: 
\begin{lemma} 
For every configuration $\mathcal{C}$, $max_{Y}(f(\mathcal{C},Y))$ is obtained for $y^{*}\in \{A,B,C,D\}$, the corners of the square in Figure~\ref{plot1}.  
\label{lemma1}
\end{lemma} 
\textbf{Proof. } A simple consequence of the fact that $f$ is a convex function  (Property 2.1, \cite{morris1988facilities}). $\Box$

\textbf{Robustness: } In the two-dimensional case we will use the decomposition in Figure~\ref{plot1} (b). Without loss of generality we assume that $f(\mathcal{C})=D$, otherwise we could transform the $x$ and $y$ coordinates separately, if needed, to make this statement true. Let $n_1,n_2,\ldots, n_4$ be respectively the number of points of $x$ in each of the squares\footnote{There is some ambiguity as to how to assign points on the boundary to a single rule. We will use the following rule: whenever we have to assign boundary points we will prefer the triangle "with lower coordinates", where coordinates are considered in the order $x_1,x_2,\ldots, x_k$. For instance points on segment $[B,R]$ could be counted  either by $n_2$ or $n_3$. We will assign them to $n_2$ since triangle BRV is "to the right of triangle BRU" with respect to first coordinate $x_1$. Similarly, points on [VR) will be counted by $n_3$ rather than $n_4$, points on segment $(RS]$ by $n_5$, rather than $n_6$ and so on.}. We call four-tuple $n=(n_1,n_2,n_3,n_4)$ \emph{the type} of configuration $\mathcal{C}$. 

Consider configuration $\mathcal{C}_1$ consisting of $n_1$ points in $T$, $n_2$ points in $R$, $n_3$ points in $S$, $n_4$ points in $D$. Clearly $\mathcal{C}_1$ has type $n$ as well, thus 
$F(\mathcal{C}_1,P_{\mathcal{C}})=(n_1+n_3)\frac{1}{2}+n_2\frac{\sqrt{2}}{2}$. We claim that
 \begin{align}
 \frac{OPT(\mathcal{C}_1)}{F(\mathcal{C}_1,P_{\mathcal{C}})}= \frac{(n_1+n_3)\sqrt{5}+(n_2+2n_4)\sqrt{2}}{(n_1+n_3)+n_2\sqrt{2}}\leq \frac{3+\lambda}{1-\lambda}
 \label{foo}
 \end{align}
 To prove this we first need to show that $OPT(\mathcal{C}_1)=SW(\mathcal{C}_1,D)$.  
 First of all, it holds that $SW(\mathcal{C}_1,C)\leq SW(\mathcal{C}_1,B)$: points $R,S$ are equally close to $B,C$, while $T,D$ are closer to $C$ than to $B$. Similarly, $SW(\mathcal{C}_1,A)\leq SW(\mathcal{C}_1,B)$. Finally $SW(\mathcal{C}_1,D)\leq SW(\mathcal{C}_1,B)$: $R$ is equally close to $B,D$, while $T,S,D$ are closer to $D$ than to $B$. Since $f(\mathcal{C}_1)=D$, we have $n_1+n_2+\lambda n \geq n_3+n_4$ and $n_2+n_3+\lambda n\geq n_1+n_4$. In particular, we infer $n_2+\lambda n\geq n_4$, i.e. 
 $n_{2}(1+\lambda)-n_{4}(1-\lambda)\geq -\lambda(n_1+n_{3})$. The last inequality in~(\ref{foo}) is equivalent to: 
\begin{align*}
(n_1+n_3)(3+\lambda - \sqrt{5}(1-\lambda))+ n_2(\sqrt{2}(3+\lambda)-\sqrt{2}(1-\lambda))-\\ - 2\sqrt{2}n_4 (1-\lambda)\geq 0,  \mbox{ i.e. } \\
(n_1+n_3)(3-\sqrt{5}+\lambda(1+\sqrt{5}))+ 2\sqrt{2}[n_2(1+\lambda)- n_4 (1-\lambda)]\geq 0.   
\label{foo2} 
\end{align*} 
Since $0\leq \lambda \leq 1$, it follows that $3- \sqrt{5} + \lambda(1+\sqrt{5}-2\sqrt{2})\geq 3-\sqrt{5}-(1+\sqrt{5}-2\sqrt{2})=2(1+\sqrt{2}>\sqrt{5})>0$. We infer that~(\ref{foo2}) (hence the last inequality in~(\ref{foo})) must be true. 

 \noindent\textbf{Proving $\mathbf{\Delta \geq 0}$:} $\Delta=F(\mathcal{C},P_{\mathcal{C}})-SW(\mathcal{C}_1,f(\mathcal{C},P_{\mathcal{C}})) - OPT(\mathcal{C})+$
  \begin{align*}
 +  SW(\mathcal{C}_{1},y^{*}) = \sum_{i=1}^{n} d(x_i,D)- \frac{n_{1}+n_{3}}{2}-n_{2}\frac{\sqrt{2}}{2} +n_{1}d(T,y^{*})+ \\ + n_{2}d(R,y^{*})+ n_{3}d(S,y^{*}) + n_{4}d(D,y^{*}) 
  \end{align*} 
 We know that $y^{*}\in \{A,B,C,D\}$. There are four cases. We treat here the case $y^{*}=B$ (the other ones are similar). Hence 
 \begin{align} 
 \Delta=\sum_{i=1}^{n} (d(x_i,D) - d(x_i,B)) + (n_{1}+n_{3})\frac{\sqrt{5}-1}{2} + n_4\sqrt{2}
 \end{align} 
 We will break $\Delta$ into eight sums, corresponding to points $x_i$ being in exactly one of the triangles $ASR, AVR, BVR, BUR, CUR, CTR$, $DTR, DSR$. First, if $x_{i}$ is in one of  $AVR,BVR,BUR,CUR$ then $d(x_{i},D)$\\ $\geq d(x_{i},B)$, so these four sums are nonnegative. Let us show that the remaining sum is $\geq 0$ as well, proving this way that $\Delta\geq 0$. 
 
 Let us consider, for instance, the points $x_i$ that belong to the triangle $CTR$. For every such $x_i$ we have $d(x_{i},D)-d(x_{i},B)\geq d(T,D)-d(T,B)= \frac{1-\sqrt{5}}{2}$, so $d(x_{i},D)-d(x_{i},B)+\frac{\sqrt{5}-1}{2}\geq 0$. Since the number of such points is at most $n_1$ and $\sqrt{5}>1$,
 \begin{align}
 \sum_{i: x_i\in CTR} (d(x_i,D) - d(x_i,B)) + n_{1}\frac{\sqrt{5}-1}{2}\geq 0. 
 \end{align}
 The cases of sums corresponding to points in triangles $ARS,DSR$, $DTR$ are similar. Adding all the eight inequalities proves that $\Delta \geq 0$. 
 
\noindent \textbf{Consistency:} Again without loss of generality, we assume that the predicted point (which also realizes the optimum) is $P_{\mathcal{C}}=D$. We will also assume that the mechanism~\ref{minustwo-alg} yields one of $A,B,C$, otherwise the consistency ratio is 1. We deal here with the case where the algorithm outputs $Z_{\mathcal{C}}=B$, the others are similar. Given our assumptions on the predictor and algorithm, the only possible case is $n_1+n_2+\lambda n< n_3+n_4$ and $n_2+n_3+\lambda n< n_1+n_4$.  So we have $n_2+\lambda n < n_4$, i.e. $n_4(1-\lambda)-n_2(1+\lambda) >\lambda(n_1+n_3) $ (**). Consider configuration $\mathcal{C}_{2}$ consisting of $n_1$ points at $U$, $n_2$ points at $B$, $n_3$ points at $V$, $n_4$ points at $R$. To prove this we first need to show that $OPT(\mathcal{C}_2)=SW(\mathcal{C}_2,D)$.  
 First, it holds that $SW(\mathcal{C}_2,C)< SW(\mathcal{C}_2,D)$: points $R,V$ are equally close to $C,D$, while $U,V$ are closer to $C$ than to $B$. Similarly, $SW(\mathcal{C}_2,A)< SW(\mathcal{C}_2,D)$. Finally $SW(\mathcal{C}_2,B)< SW(\mathcal{C}_2,D)$: $R$ is equally close to $B,D$, while $U,V,B$ are closer to $B$ than to $D$. Second, $F(\mathcal{C}_2,P_{\mathcal C})=SW(\mathcal{C}_2,B)$. Thus
\begin{align}
\frac{OPT(\mathcal{C}_2)}{F(\mathcal{C}_2,P_{\mathcal C})}= \frac{(n_1+n_3)\sqrt{5}+(n_4+2n_2)\sqrt{2}}{(n_1+n_3)+n_4\sqrt{2}} 
<   \frac{3-\lambda}{1+\lambda} 
\label{foo3} 
\end{align}
Indeed, the last inequality is equivalent to $(n_1+n_3)(3-\lambda - \sqrt{5}(1+\lambda))+2\sqrt{2}[n_4(1-\lambda)-n_{2}(1+\lambda)]> 0$. Applying (**) we get 
\begin{align*}
(n_1+n_3)(3-\sqrt{5}-\lambda(1+\sqrt{5}))+ 2\sqrt{2}[n_4(1-\lambda)-n_{2}(1+\lambda)] > (n_1 \\
+n_3)[3-\sqrt{5}- \lambda(1+\sqrt{5}-2\sqrt{2})]\geq 2(n_1+n_3)[1-\sqrt{5}+\sqrt{2}]\geq 0. 
\end{align*} 
thus proving~(\ref{foo3}) 
(we have used inequalities, $1+\sqrt{5}-2\sqrt{2}>0$, $0\leq \lambda\leq 1$ and $1-\sqrt{5}+\sqrt{2}>0$). 
To complete the proof of the consistency bound all we have to prove is that 
\begin{align}
\frac{OPT(\mathcal{C})}{F(\mathcal{C},P_{\mathcal C})} \leq \frac{SW(\mathcal{C}_2,y^{*})}{F(\mathcal{C}_2,P_{\mathcal{C}})} \leq \frac{OPT(\mathcal{C}_2)}{F(\mathcal{C}_2,P_{\mathcal C})}
\label{foo4} 
\end{align} 
Similarly to the proof of robustness, to prove inequality~(\ref{foo4}) we consider quantity $\Delta_2:= F(\mathcal{C},P_{\mathcal{C}})- F(\mathcal{C}_2,P_{\mathcal{C}})-OPT(\mathcal{C})+SW(\mathcal{C}_{2},y^{*})$, and aim to prove that $\Delta_2\geq 0$. 

\noindent\textbf{Proving $\mathbf{\Delta_2 \geq 0}$:} $\Delta_2=F(\mathcal{C},P_{\mathcal{C}})-F(\mathcal{C}_2,P_{\mathcal{C}})-OPT(\mathcal{C})+SW(\mathcal{C}_{2},y^{*})$
  \begin{align*}
 = \sum_{i=1}^{n} d(x_i,B)-  \frac{n_{1}+n_{3}}{2}-n_{4}\frac{\sqrt{2}}{2} - \sum_{i=1}^{n} d(x_i,y^{*})
 + n_{1}d(U,y^{*})+ \\ + n_{2}d(B,y^{*})+ n_{3}d(V,y^{*})+n_{4}d(R,y^{*})
 \end{align*} 
 We know that $y^{*}\in \{A,B,C,D\}$. There are four cases. We treat here the case $y^{*}=D$ (the other ones are similar). Hence 
\begin{align} 
 \Delta_2=\sum_{i=1}^{n} (d(x_i,B) - d(x_i,D)) + (n_{1}+n_{3})\frac{\sqrt{5}-1}{2} + n_2\sqrt{2}
 \end{align} 
 The proof that $\Delta_2\geq 0$ uses the same idea of triangle decomposition as the corresponding proof for $\Delta \geq 0$. 

\noindent\textbf{Strategyproofness:}  Let $\mathcal{C}=(p_{1},p_{2},\ldots, p_n)$ be an arbitrary configuration, and without loss of generality assume that $f(\mathcal{C})=D$. Let $i\in [n]$ be an agent, and assume that by misreporting location as $q_i$, agents $i$ is able to strictly improve its utility. This means that in the new profile $\mathcal{C}^{\prime}$, $f(\mathcal{C}^{\prime})\neq f(\mathcal{C})$. 
There are two cases: 
\begin{itemize}[leftmargin=*]
\item $\mathbf{f(\mathcal{C}^{\prime})\in \{A,C\}.}$ These two cases are symmetric and without loss of generality we can assume that $f(\mathcal{C}^{\prime})=A$. This means that $n_1+n_2+\lambda n\geq n_3+n_4$ and $n_2+n_3+\lambda n\geq n_1+n_4$ but the corresponding inequalities for $n_1^{\prime}, n_2^{\prime}, n_3^{\prime}, n_4^{\prime}$ read 
\begin{align}
n_1^{\prime}+n_2^{\prime}+\lambda n\geq n_3^{\prime}+n_4^{\prime} \mbox{ but } n_2^{\prime}+n_3^{\prime}+\lambda n< n_1^{\prime}+n_4^{\prime}.
\end{align}
In particular $n_2+n_3\geq n\frac{1-\lambda}{2}\mbox{ but } n_2^{\prime}+n_3^{\prime}< n\frac{1-\lambda}{2}$. 
No point whose location is in the rectangle $BUSA$ gains by misreporting it ($D$ is at least as far as $A$). So i does not belong to rectangle $BUSA$. But this implies $n_2+n_3\leq n_{2}^{\prime}+n_{3}^\prime$, a contradiction. 
\item $\mathbf{f(\mathcal{C}^{\prime})=B.}$ Then $n_1+n_2+\lambda n\geq n_3+n_4$ and $n_2+n_3+\lambda n\geq n_1+n_4$, but $n_1^{\prime}+n_2^{\prime}+\lambda n< n_3^{\prime}+n_4^{\prime}$ and $n_2^{\prime}+n_3^{\prime}+\lambda n< n_1^{\prime}+n_4^{\prime}$. Hence $n_1+n_2\geq n\frac{1-\lambda}{2}>n_1^{\prime}+n_2^{\prime}$ and $n_2+n_3\geq n\frac{1-\lambda}{2}>n_2^{\prime}+n_3^{\prime}$. 

Since misreporting a single point decreases both sums $n_2+n_3$ and $n_1+n_2$, we infer the fact that agent $i$ is in the rectangle $BURV$ and $n_1+n_2=n_3+n_4$, $n_2+n_3=n_1+n_4$. But this is a contradiction: no point whose true location is in the rectangle $BURV$ gains by misreporting its location ($D$ is already their farthest point). So $f$ is strategyproof. 

A counterexample that shows that a coalition of two agents can manipulate Mechanism~\ref{minustwo-alg} is presented in Figure~\ref{manipulation-2}: $\mathcal{C}$ consists of 3 agents at $T,S$ each, one at each of $R,D$. Clearly $f(\mathcal{C})=D$. On the other hand, making one agent from each of $T,S$ misreport their location as $D$ yields a configuration $\mathcal{C}^{\prime}$ such that $f(\mathcal{C}^{\prime})=B$. Both agents benefit from misreporting. 

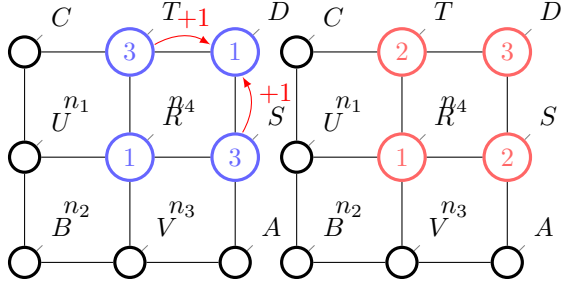
\begin{figure}
\begin{center}
\begin{minipage}{.25\textwidth} 
\begin{tikzpicture}[scale=0.7][
roundnode/.style={circle, draw=green!60, fill=green!5, very thick, minimum size=7mm},
squarednode/.style={rectangle, draw=red!60, fill=red!5, very thick, minimum size=5mm},
]
\draw (1,1) node{$n_2$};
\draw (3,1) node{$n_3$};
\draw (3,3) node{$n_4$};
\draw (1,3) node{$n_1$};
\node [circle,draw,minimum size=0.6cm,color=blue!60,very thick,pin={[pin distance=0.1cm]45:{$D$}}](d) at (4,4){1};
\node [circle,draw,minimum size=0.6cm,color=blue!60,very thick,pin={[pin distance=0.1cm]45:{$R$}}](r) at (2,2){1};
\node [circle,draw,minimum size=0.6cm,color=blue!60,very thick,pin={[pin distance=0.1cm]45:{$T$}}](t) at (2,4){3};
\node [circle,draw,minimum size=0.6cm,color=blue!60,very thick,pin={[pin distance=0.1cm]45:{$S$}}](s) at (4,2){3};
\node [circle,draw,minimum size=0.4cm,color=black,very thick,pin={[pin distance=0.1cm]45:{$U$}}](u) at (0,2){};
\node [circle,draw,minimum size=0.4cm,color=black,very thick,pin={[pin distance=0.1cm]45:{$V$}}](v) at (2,0){};
\node [circle,draw,minimum size=0.4cm,color=black,very thick,pin={[pin distance=0.1cm]45:{$A$}}](a) at (4,0){};
\node [circle,draw,minimum size=0.4cm,color=black,very thick,pin={[pin distance=0.1cm]45:{$B$}}](b) at (0,0){};
\node [circle,draw,minimum size=0.4cm,color=black,very thick,pin={[pin distance=0.1cm]45:{$C$}}](c) at (0,4){};
\draw[-] (d.south) -- (s.north);
\draw[-] (t.south) -- (r.north);
\draw[-] (c.south) -- (u.north);
\draw[-] (u.south) -- (b.north);
\draw[-] (r.south) -- (v.north);
\draw[-] (s.south) -- (a.north);
\draw[-] (c.east) -- (t.west);
\draw[-] (u.east) -- (r.west);
\draw[-] (b.east) -- (v.west);
\draw[-] (t.east) -- (d.west);
\draw[-] (r.east) -- (s.west);
\draw[-] (v.east) -- (a.west);
\draw[-latex,red] ([yshift=0.4em]t.east) to [bend left] node[pos=0.7,above] {+1} ([yshift=0.4em]d.west);
\draw[-latex,red] ([xshift=0.4em]s.north) to [bend right] node[pos=0.7,right] {+1} ([xshift=0.4em]d.south);
\end{tikzpicture}
\end{minipage}
~~~
\begin{minipage}{.25\textwidth}
\begin{tikzpicture}[scale=0.7][
roundnode/.style={circle, draw=green!60, fill=green!5, very thick, minimum size=7mm},
squarednode/.style={rectangle, draw=red!60, fill=red!5, very thick, minimum size=5mm},
]

\draw (1,1) node{$n_2$};
\draw (3,1) node{$n_3$};
\draw (3,3) node{$n_4$};
\draw (1,3) node{$n_1$};
\node [circle,draw,minimum size=0.6cm,color=red!60,very thick,pin={[pin distance=0.1cm]45:{$D$}}](d) at (4,4){3};
\node [circle,draw,minimum size=0.6cm,color=red!60,very thick,pin={[pin distance=0.1cm]45:{$R$}}](r) at (2,2){1};
\node [circle,draw,minimum size=0.6cm,color=red!60,very thick,pin={[pin distance=0.1cm]45:{$T$}}](t) at (2,4){2};
\node [circle,draw,minimum size=0.6cm,color=red!60,very thick,pin={[pin distance=0.1cm]45:{$S$}}](s) at (4,2){2};
\node [circle,draw,minimum size=0.4cm,color=black,very thick,pin={[pin distance=0.1cm]45:{$U$}}](u) at (0,2){};
\node [circle,draw,minimum size=0.4cm,color=black,very thick,pin={[pin distance=0.1cm]45:{$V$}}](v) at (2,0){};
\node [circle,draw,minimum size=0.4cm,color=black,very thick,pin={[pin distance=0.1cm]45:{$A$}}](a) at (4,0){};
\node [circle,draw,minimum size=0.4cm,color=black,very thick,pin={[pin distance=0.1cm]45:{$B$}}](b) at (0,0){};
\node [circle,draw,minimum size=0.4cm,color=black,very thick,pin={[pin distance=0.1cm]45:{$C$}}](c) at (0,4){};
\draw[-] (d.south) -- (s.north);
\draw[-] (t.south) -- (r.north);
\draw[-] (c.south) -- (u.north);
\draw[-] (u.south) -- (b.north);
\draw[-] (r.south) -- (v.north);
\draw[-] (s.south) -- (a.north);
\draw[-] (c.east) -- (t.west);
\draw[-] (u.east) -- (r.west);
\draw[-] (b.east) -- (v.west);
\draw[-] (t.east) -- (d.west);
\draw[-] (r.east) -- (s.west);
\draw[-] (v.east) -- (a.west);

\end{tikzpicture}
\end{minipage}

\end{center} 
\caption{Manipulation by a coalition of size 2.}
\label{manipulation-2} 
\end{figure} 

\end{itemize}

\subsection{Proof of Theorem~\ref{upper-circle}}

\noindent\textbf{Robustness:} Without loss of generality let us choose coordinates such that point $P_{\mathcal{C}}$ has coordinate 1/4. 

We first deal with \textbf{the case $\mathbf{f(\mathcal{C},P_{\mathcal{C}})=P_{\mathcal{C}}}$}, that is $n_P\leq n_Q+\lambda n$ or, equivalently, by substituting $n_Q=n-n_P$, $n_P\leq n(1+\lambda)/2$. 

Let $\mathcal{C}_1$ be the configuration consisting of $n_Q$ points at $P_{\mathcal{C}}-1/4$, together with $n_P$ points at $P_{\mathcal{C}}$ (Figure~\ref{cilcles_C1_C2}). The algorithm will output point $P_{\mathcal{C}}$. Hence $SW(\mathcal{C}_1,f(\mathcal{C},P_{\mathcal{C}}))= n_Q/4.$

\begin{figure}

\begin{minipage}{.25\textwidth}
\begin{tikzpicture}[scale=0.10,font=\sffamily,every path/.style={>=latex},every node/.style={draw,circle,inner sep=0.7pt,font=\small\sffamily}]
  
    \node [circle,draw,minimum size=2cm,color=red!60,very thick,pin=right:$P_C$,pin=left:Q,pin=above:{1, 0}](a){};
    \node[draw,circle,minimum size=4pt,fill,pin={260:$n_Q$}] at (0,10cm) {};
    \node[draw,circle,minimum size=4pt,fill,pin={280:$n_P$}] at (10cm,0) {};
      
\end{tikzpicture}
\end{minipage}
~~~
\begin{minipage}{.25\textwidth}
\begin{tikzpicture}[scale=0.10,font=\sffamily,every path/.style={>=latex},every node/.style={draw,circle,inner sep=0.7pt,font=\small\sffamily}]
  
    \node [circle,draw,minimum size=2cm,color=red!60,very thick,pin=right:$P_C$,pin=left:Q,pin=above:{1, 0}](a){};
    \node[draw,circle,minimum size=4pt,fill,pin={260:$n_Q$}] at (-10cm,0cm) {};
    \node[draw,circle,minimum size=4pt,fill,pin={80:$n_P$}] at (0,-10cm) {};
      
\end{tikzpicture}
\end{minipage}

\caption{(a). Configuration $\mathcal{C}_1$ (b). Configuration $\mathcal{C}_2$. }
\label{cilcles_C1_C2}
\end{figure}
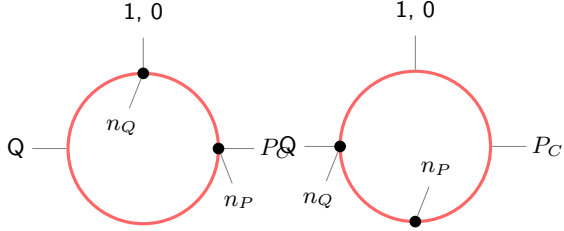

Let us find $x\in [0,1] $ such that $SW(\mathcal{C}_1,x)$ is maximized. On each of the segments $[0,1/4],[0,1/4],[1/2,3/4],[3/4,1]$ the objective function is a linear function in $x$ (generally a different one from segment to segment).  Hence the maximum is reached at one of the points $x=0,x=1/4,x=1/2,x=3/4$. These values are, respectively, $n_{P}/4,n_{Q}/4,n_{P}/2+n_{Q}/4,n_{P}/4+n_{Q}/2$. We infer that $OPT(\mathcal{C}_1)=\frac{n+max(n_P,n_Q)}{4}$, reached for $y^{*}=Q=3/4$ (if $n_P\geq n_Q$) and $y^{*}=1/2$ (if $n_P<n_Q$). If $n_P\leq n_Q$ then the fraction $\frac{OPT(\mathcal{C}_1)}{SW(\mathcal{C}_1,f(\mathcal{C},P_{\mathcal{C}}))}$ is at most 3. Otherwise it is  $\frac{n_Q+2n_P}{n_Q}$. Hence 
\begin{align}
\frac{OPT(\mathcal{C}_1)}{SW(\mathcal{C}_1,f(\mathcal{C},P_{\mathcal{C}}))}\leq max(3,1+2\frac{1+\lambda}{1-\lambda})=
\frac{3+\lambda}{1-\lambda}
\end{align} 
\textbf{$\mathbf{\Delta \geq 0}$:} $\Delta = F(\mathcal{C},P_{\mathcal{C}})-SW(\mathcal{C}_1,f(\mathcal{C},P_{\mathcal{C}}))- SW(\mathcal{C},y^{*})+SW(\mathcal{C}_1,y^{*})$ 
\begin{align*}
+SW(\mathcal{C}_1,y^{*})= \sum_{i=1}^{n} d(x_i,P_{\mathcal{C}}) - \frac{n_Q}{4} -\sum_{i=1}^{n} d(x_i,Q) + \frac{n_Q+2n_P}{4} =   \\ \frac{n_P}{2} + \sum_{i=1}^{n} (d(x_i,\widehat{P}) - d(x_i,\widehat{Q})) =  \sum_{x_i\in [1/2,1]}(d(x_i,P_{\mathcal{C}}) - d(x_i,Q))\\ + \sum_{x_i\in [0,1/2]} (d(x_i,P_{\mathcal{C}}) - d(x_i,Q))+  n_P/2. 
\end{align*}
The first sum consists of non-negative terms only, since for $z\in [1/2,1]$, $d(z,P_{\mathcal{C}})\geq d(z,Q)$. For the remaining terms, when $z\in [0,1/2]$ we claim that $d(z,P_{\mathcal{C}})-d(z,Q)\geq -1/2$. This is obvious since $d(z,P_{\mathcal{C}})\geq 0$ and $d(z,Q)\leq 1/2$, with equality when $z=P_{\mathcal{C}}$. Since the number of points in $[0,1/2]$ is $n_{P}$, $\Delta \geq 0$ follows. 

The  \textbf{case $\mathbf{f(\mathcal{C},P_{\mathcal{C}})=Q}$}, that is $n_P> n_Q+\lambda n$ or, equivalently, $n_P (1-\lambda)> n_Q(1+\lambda)$, is similar: let $\mathcal{C}_2$ be the configuration consisting of $n_P$ points at $P_{\mathcal{C}}+1/4$, together with $n_Q$ points at $Q$. The algorithm will output point $Q$, with social welfare $SW(\mathcal{C}_2,f(\mathcal{C},P_{\mathcal{C}}))=n_P/4.$ 
A similar analysis to that of the first case yields $OPT(\mathcal{C}_2)= (n+max(n_P,n_Q))/4=n+n_P=(2n_P+n_Q)/4$. 
Hence 
\begin{align}
\frac{OPT(\mathcal{C}_2)}{SW(\mathcal{C}_2,f(\mathcal{C},P_{\mathcal{C}}))}=2+\frac{n_Q}{n_P}< 2+\frac{1-\lambda}{1+\lambda}=\frac{3+\lambda}{1+\lambda} \leq \frac{3+\lambda}{1-\lambda}
\end{align}
Proving $\Delta \geq 0$ follows a similar pattern to that of the first case. 

\noindent{\textbf{Consistency:}} Without loss of generality assume that $P_{\mathcal{C}}=1/4$. If $f(\mathcal{C},P_{\mathcal{C}})=P_{\mathcal{C}}$ then the consistency ratio is one. 
Assume thus that $f(\mathcal{C},P_{\mathcal{C}})=Q$ but $OPT(\mathcal{C})$ is reached at $P_{\mathcal{C}}$. By the definition of the algorithm, the only possible case is $n_P>n_Q+\lambda n$, i.e. $\frac{n_Q}{n_P}<\frac{1-\lambda}{1+\lambda}$. 
Let $\mathcal{C}_1$ be the configuration consisting of the $n_P$ points of $\mathcal{C}$, together with $n_Q$ points at $Q$. Since for every $y$ we have $d(P,y)+d(Q,y)=1/2$, it follows that the optimum social welfare of $\mathcal{C}_1$ is still reached at $P_{\mathcal{C}}$, and is equal to $\Sigma+n_{Q}/2$, where $\Sigma:=\sum_{x_i\in (0,1/2)} |x_i-1/4|$. On the other hand $SW(\mathcal{C}_1,P_{\mathcal{C}})=\sum_{x_{i}\in (0,1/2)} 1/2-|x_i-1/4|= n_{P}/2-\Sigma$. Since $\Sigma \leq n_P/4$, the approximation ratio is
\[
\frac{OPT(\mathcal{C}_1)}{F(\mathcal{C}_1,P_{\mathcal{C}})}=\frac{\Sigma + n_Q/2}{n_P/2-\Sigma}\leq \frac{\frac{n_P}{4}+\frac{n_Q}{2}}{\frac{n_P}{2}-\frac{n_{P}}{4}}=1+\frac{2n_Q}{n_P}\leq 1+2\frac{1-\lambda}{1+\lambda}= \frac{3-\lambda}{1+\lambda}. 
\]
But $OPT(\mathcal{C})\leq OPT(\mathcal{C}_1)$ and $F(\mathcal{C}_1,P_{\mathcal{C}})\geq F(\mathcal{C},P_{\mathcal{C}})$. The consistency upperbound follows. 

\noindent{\textbf{Group strategyproofness:}} 
Assume first that $n_P\leq n_Q+\lambda n$, i.e. Algorithm~\ref{minustwo-alg} outputs point $P_{\mathcal{C}}$. Let $S\subset [n]$ be a set of agents. If $S$ contained a point $W$ such that $d(W,P_{\mathcal{C}})\geq d(W,Q)$ then $W$ cannot increase its utility by misreporting location. So all points in $S$ 
 are farther to $Q$ than to $P$. This means that all points in $S$ are counted into $n_P$ and by misreporting their location they will lead to point $P_{\mathcal{C}}$ being chosen as well. 
The case $n_P> n_Q+\lambda n$ is similar. 

\subsection{Proof of Theorem \ref{upper-trees}} 

First, a well-known problem in algorithms, a method to find diameters in trees, implies the fact that the algorithm is correct: if $a$ is the farthest from $P_\mathcal{C}$ then $P$ is a peripheral point: the farthest point $b$ from $a$ induces a diameter $ab$  in $T^{\prime}$. Moreover, by the definition of points $a,b$, we have $d(P_{\mathcal{C}},b)\leq d(P_{\mathcal{C}},a)$. 

\noindent \textbf{Group strategyproofness: } This is similar to the corresponding proof in \cite{mei2018mechanism}, and is based on the fact that points in $T_a$ are closer to $a$ than to $b$, while points in $T_b$ are closer to $b$ than to $a$. 

\noindent \textbf{Robustness: }  We first consider the case $n_1+\lambda n> n_2$, that is the algorithm outputs point $b$. 
Let $y^{*}$ be the optimal location for configuration $\mathcal{C}$ and $u^{*}$ the projection of $y^{*}$ on $[a,b]$.  Consider configuration $\mathcal{C}_1$ in which there are $n_1$ points at $m_{a,b}$ and $n_2$ points at $b$. $SW(\mathcal{C}_1,y^*)=n_1\cdot d(m_{a,b},y^{*}) + n_2\cdot d(y^{*},b)\leq n_1\cdot 
\frac{d(a,b)}{2}+n_{2}d(a,b)< n_{1}d(a,b)[\frac{1}{2}+\frac{1+\lambda}{1-\lambda}]=n_1\cdot d(a,b)\cdot \frac{3+\lambda}{2(1-\lambda)}$. On the other hand $F(\mathcal{C}_1,P_C)=n_1\cdot d(m_{a,b},b)+n_2\cdot d(b,b)=n_1\frac{d(a,b)}{2}$. So 
\[
\frac{SW(\mathcal{C}_1,y^{*})}{F(\mathcal{C}_1,P_C)}<\frac{3+\lambda}{1-\lambda}
\]
\textbf{$\Delta \geq 0$:} $\Delta = F(\mathcal{C},P_{\mathcal{C}})-SW(\mathcal{C}_1,f(\mathcal{C},P_{\mathcal{C}})) - OPT(\mathcal{C})+SW(\mathcal{C}_{1},y^{*})=$
  \begin{align*}
= \sum_{x_i\in \mathcal{C}} d(x_i,b)- n_1 d(m_{a,b},b) - \sum_{x_i\in \mathcal{C}} d(x_i,y^{*})  +  n_1 d(m_{a,b},y^*)\\ + n_2 d(b,y^{*})=   \sum_{x_i\in T_a} (d(x_i,b)-d(m_{a,b},b) - d(x_i,w^{*}) +d(m_{a,b},w^{*})) \\ + \sum_{x_i\in T_b} (d(x_i,b)-d(x_i,w^{*}) + d(b,w^{*}))  =  \sum_{x_i\in T_a} (d(m_{a,b},w^{*})\\ +  d(m_{a,b},x_i)- 
d(x_i,w^{*})) + \sum_{x_i\in T_b} (d(x_i,b)+ d(b,w^{*})
 -d(x_i,w^{*}) \geq 0 \end{align*} 
 by triangle inequalities. 
 
 Consider now the case $n_1+\lambda n\leq n_2$, that is the algorithm chooses the point $a$. Consider the configuration $\mathcal{C}_2$ in which there are $n_1$ points at $a$ and $n_2$ points at $m_{a,b}$.  $SW(\mathcal{C}_2,y^*)=n_1\cdot d(a,y^{*}) + n_2\cdot d(y^{*},m_{a,b})\leq n_1\cdot 
d(a,b)+n_{2}d(a,b)/2 \leq n_{2}d(a,b)[\frac{1}{2}+\frac{1-\lambda}{1+\lambda}]=n_1\cdot d(a,b)\cdot \frac{3-\lambda}{2(1+\lambda)}$. On the other hand $F(\mathcal{C}_2,P_C)=n_1\cdot d(a,a)+n_2\cdot d(m_{a,b},a)=n_2\frac{d(a,b)}{2}$. So 
\begin{equation} 
\frac{SW(\mathcal{C}_2,y^{*})}{F(\mathcal{C}_2,P_C)}\leq \frac{3-\lambda}{1+\lambda}\leq \frac{3+\lambda}{1-\lambda}
\label{foo1} 
\end{equation} 
\textbf{$\Delta \geq 0$:} $\Delta = F(\mathcal{C},P_{\mathcal{C}})-SW(\mathcal{C}_1,f(\mathcal{C},P_{\mathcal{C}})) - OPT(\mathcal{C})+SW(\mathcal{C}_{1},y^{*})$
  \begin{align*}
= \sum_{x_i\in \mathcal{C}} d(x_i,a)- n_2 d(m_{a,b},a) - \sum_{x_i\in \mathcal{C}} d(x_i,y^{*})  +  n_1 d(a,y^*)+ \\ + n_2 d(m_{a,b},y^{*})=   \sum_{x_i\in T_a} (d(x_i,a)- d(x_i,w^{*}) +d(a,w^{*})) \\  + \sum_{x_i\in T_b} (d(x_i,a)-d(m_{a,b},a) - d(x_i,w^{*}) + d(m_{a,b},w^{*}))  =  \\ = \sum_{x_i\in T_a} d(a,w^{*}) +  d(a,x_i)- 
d(x_i,w^{*}) +\\ +  \sum_{x_i\in T_b} (d(x_i,m_{a,b}) + d(m_{a,b},w^{*})
 -d(x_i,w^{*}) \geq 0 \end{align*} 
 by triangle inequalities. 
  
\noindent\textbf{Consistency: }  Assume that $P_{\mathcal{C}}=y^{*}$ is a peripheral point. There are two cases: 

\noindent\textbf{Case 1: $\mathbf{n_1+\lambda n\leq n_2}$.} That is the algorithm outputs point $a$. 
We actually obtained the desired bound in the course of proving robustness (first inequality in equation~(\ref{foo1})). 

\noindent\textbf{Case 1: $\mathbf{n_1+\lambda n> n_2}$.} That is the algorithm outputs point $b$. But by the definition of the algorithm (the fact that we take $b=P_{\mathcal{C}}$ if possible) it follows that $P_{\mathcal{C}}=y^{*}=b$. The approximation ratio is 1. 

\section{Proof of Theorem~\ref{thm:lb}}
\label{sec:lb} 

We will use characterization of strategyproof mechanisms from \cite{han2012moneyless, ibara2012characterizing}. They prove that strategyproof mechanisms have the following form: there exist $0\leq r<s\leq 1$ and $0\leq k\leq n$ s.t., if $v=\frac{r+s}{2}$ then 
\[
T_{k}^{r,s}(s_1,s_2,\ldots s_n)=\left\{ \begin{array}{ll}
r, & \mbox{ if }|\{i: s_{i}\leq v\}|\leq k, \\
s, & \mbox{ otherwise.} \\
\end{array}
\right. 
\]
In our cases clearly $r=0,s=1$, hence $v=1/2$. 
We will compute next the approximation ratios for the threshold mechanism $T_k^{0,1}$ under each of the hypotheses $y^{*}=0$, $y^{*}=1$. 

Clearly $OPT(\mathcal{C})=max(\sum_{i} x_{i}, \sum_{i=1}^{n} (1-x_{i}))$. To find the approximation ratios of $T_k^{0,1}$ we will look for configurations $\mathcal{C}$ for which the optimum is realized at $y^{*}=0$ while $T_{k}(\mathcal{C})=1$ (or viceversa).  

Indeed, to realize the optimum at $y^{*}=0$ while $T_{k}(\mathcal{C})=1$ we need to maximize $\sum_{i} x_i$ while making sure that at least $k+1$ points of $\mathcal{C}$ are in $[0,1/2]$. This is achieved by taking $\mathcal{C}_{1}$ to be the configuration consisting of $k+1$ points at 1/2 and $n-k-1$ points at 1. 
It can be verified that the optimum for $\mathcal{C}_1$ is indeed realized at $y^{*}=0$, since $\sum_{i} x_{i}=k/2+n-k=n-k/2$, while $\sum_{i} (1-x_{i})=k/2$. So 
\[
\frac{OPT(\mathcal{C}_1)}{T_{k}(\mathcal{C}_1)}= \frac{n-k/2}{k/2}=\frac{2n}{k}-1.
\]
On the other hand, to realize the optimum at $Y^{*}=1$ while $T_{k}(\mathcal{C})=0$ we need to minimize $\sum_{i} x_i$ while making sure that at most $k$ points of $\mathcal{C}$ are in $[0,1/2]$. Let $\mathcal{C}_{2}$ the configuration consisting of $k$ points at 0 and $n-k$ points at $1/2+\epsilon$. It can be verified that for small enough $\epsilon >0$ the optimum for $\mathcal{C}_2$ is indeed realized at $P=1$ since $\sum_{i} x_{i}=(n-k)(1/2+\epsilon)\approx \frac{n-k}{2}$, while $\sum_{i} (1-x_{i})=k+(n-k)(1/2-\epsilon)\approx (n-k)/2+k$. 
\[
\frac{OPT(\mathcal{C}_2)}{T_{k}(\mathcal{C}_2)}= \frac{(n+k)/2-\epsilon(n-k)}{(n-k)/2+(n-k)\epsilon}\leq \frac{n+k}{n-k}=\frac{2n}{n-k}-1. 
\]
In conclusion, for $k\leq n/2$ $T_k$ has approximation ratio $\frac{2n}{k}-1$, while for $k>n/2$ $T_k$ has approximation ratio $\frac{2n}{n-k}-1$. 

Consider now a mechanism with predictions $f:X^{n}\times \{0,1\}\rightarrow \{0,1\}$ and its restrictions $f_0,f_1$ to each of the cases $P_{\mathcal{C}}=0$, $P_{\mathcal{C}}=1$. Mechanisms $f_0,f_1$ are strategyproof. Since $f$ is $(1+c)$-consistent, for every $\mathcal{C}$ such that $OPT(\mathcal{C})=0$, $\frac{OPT(\mathcal{C})}{SW(\mathcal{C},f_0(\mathcal{C}))}\leq 1+c$, and similarly, for every $\mathcal{C}$ s.t. $OPT(\mathcal{C})=1$, $\frac{OPT(\mathcal{C})}{SW(\mathcal{C},f_1(\mathcal{C}))}\leq 1+c$.

By the previous discussion it follows that there exist numbers $k_0,k_1$ such that $f_0\equiv T_{k_0}$ and $f_1\equiv T_{k_1}$. Furthermore, since $f$ is $(1+c)$-consistent, we infer that $\frac{2n}{k_0}-1\leq 1+c$ or, equivalently, $k_0\geq \frac{2n}{2+c}$. Similarly, $\frac{2n}{n-k_1}-1\leq 1+c$, hence $k_1\leq n-\frac{2n}{2+c}=\frac{2cn}{2+c}$. 

Now let's consider the case of wrong predictions: suppose $P_{\mathcal{C}}=0$ but, in fact, $y^{*}=1$. 
This means that $f$ acts on $\mathcal{C}$ as $T_{k_0}$ would do. Consider again configuration $\mathcal{C}_2$ consisting of $k_0$ points at $0$, the rest at $1/2+\epsilon$. Clearly the optimum of this configuration is obtained at $y^{*}=1$. On the other hand $
\frac{OPT(\mathcal{C}_2)}{F(\mathcal{C}_2,P_{C_2})}= \frac{(n+k_0)/2-\epsilon(n-k_0)}{(n-k_0)/2+(n-k_0)\epsilon}$. Since $k_0\geq \frac{2n}{2+c}$, $
\frac{OPT(\mathcal{C}_2)}{F(\mathcal{C}_2,P_{C_2})}\geq \frac{\frac{n}{2}(1-\epsilon)+\frac{2n}{2+c}(\frac{1}{2}+\epsilon)}{\frac{n}{2}(1+2\epsilon)-\frac{2n}{2+c}(1/2+\epsilon)}=\frac{\frac{1}{2}+\frac{1}{(2+c)}+\epsilon\frac{2-c}{2(2+c)}}{\frac{1}{2}-\frac{1}{2+c}+\epsilon\frac{c}{2+c}}$. 
The limit (as $\epsilon \rightarrow 0$) of the right-hand side is $
\frac{\frac{1}{2}+\frac{1}{(2+c)}}{\frac{1}{2}-\frac{1}{2+c}}=\frac{4+c}{c}=1+\frac{4}{c}$. 
So for any $\delta >0$ we can find $\epsilon > 0$ such that configuration $C_2$ witnesses the fact that the mechanism $f$ is \textbf{not} $(1+\frac{4}{c}+\delta)$-robust. 

\section{Proof of Theorem~\ref{upper-dual}} 

\noindent\textbf{$\gamma$-group strategyproofness:} The proof is very similar to the corresponding proof in \cite{zou2015facility} (and our own arguments for Theorem~\ref{upper-1-2}), and is omitted from this extended abstract. 
We only highlight the fact that (because the intervals we consider are $[0,\frac{1-\lambda}{2}),[1-\lambda/2,1]$, which have unequal lengths), for $\lambda >0$ the mechanism cannot be group-strategyproof, otherwise it would contradict the characterization \cite{han2012moneyless} of strategyproof mechanisms for agents of type 0. 

\noindent\textbf{Robustness: }  The result will follow from the following lemmas: 
\begin{lemma} 
Let $\mathcal{C}$ be a configuration with all agents of type 1. Then 
\begin{equation} 
\frac{OPT(\mathcal{C})}{f(\mathcal{C},P_{\mathcal{C}})}\leq \frac{3+\lambda}{1-\lambda}
\end{equation} 
\end{lemma} 
\begin{proof} 
Without loss of generality assume that $P_{\mathcal{C}}=1$ (if this is not the case, simply switch labels 1 and 0). 
Let $y^{*}$ be the point realizing the optimum social welfare for $\mathcal{C}$. $y^{*}$ can be characterized as the point that maximizes $\sum_{i=1}^{n} (1-d(x_i,y^{*}))$ (that is, it minimizes 
$\sum_{i=1}^{n} d(x_i,y^{*})$. It is well-known (e.g. \cite{procaccia2013approximate}) that $y^{*}$ is the median point of configuration $\mathcal{C}$ (any point between the two median points, if $n$ is even). That is: let $\mathcal{C}=x_1\leq x_2\leq \ldots \leq x_n$. If $n=2k+1$ then $y^{*}=x_{k+1}$. Otherwise $y^{*}\in [x_k,x_{k+1}]$ and we will take $y^{*}=\frac{x_k+x_{k+1}}{2}$. 

Assume first  that $\mathbf{n_1\leq n_2}$ (i.e. Algorithm~\ref{minusone-alg} outputs $P_{\mathcal{C}}=1$). Suppose first that $n=2k$. Then, by condition $\mathbf{n_1\leq n_2}$ it follows that there are at least $k$ members of $\mathcal{C}$ in the interval $[0,\frac{1+\lambda}{2}]$. So $x_k\leq \frac{1+\lambda}{2}$. 
For $1\leq r\leq k$ we have: 
\begin{align*}
\frac{u_{r}(y^{*})+u_{n+1-r}(y^{*})}{u_{r}(f(\mathcal{C},P_{\mathcal{C}}))+u_{n+1-r}(f(\mathcal{C},P_{\mathcal{C}}))} = \frac{2-|x_r-y^{*}|-|x_{n+1-r}-y^{*}|}{2-x_r-x_{n+1-r}} \\ 
= \frac{2-|x_{n+1-r}-x_r|}{2-2x_r-|x_{n+1-r}-x_{r}|}\leq \frac{2-(1-x_r)}{2-2x_r-(1-x_r)}= \\ = \frac{1+x_r}{1-x_{r}}\leq  \frac{1+\frac{1+\lambda}{2}}{1-\frac{1+\lambda}{2}} = \frac{3+\lambda}{1-\lambda}
\end{align*} 
In the previous  inequalities we have used the fact that $x_{r}\leq y^{*}\leq x_{n+1-r}$, that 
$x_{n+1-r}-x_{r}\leq 1-x_r$, that $x_r\leq \frac{1+\lambda}{2}$. 
In conclusion 
\begin{align*} 
u_{r}(y^{*})+u_{n+1-r}(y^{*}) \leq  \frac{3+\lambda}{1-\lambda}\cdot (u_{r}(f(\mathcal{C},P_{\mathcal{C}}))+u_{n+1-r}(f(\mathcal{C},P_{\mathcal{C}})))
\end{align*} 
Summing up all these inequalities for $r=1,\ldots, k$ we get $OPT(\mathcal{C})\leq \frac{3+\lambda}{1-\lambda}\cdot F(\mathcal{C},P_{\mathcal{C}})$, which is what we wanted to prove. 
The proof is similar for the case $n=2k+1$. In this case we sum again inequalities for $i=1,\ldots, k$ and use the fact that $x_{k+1}=y^{*}$, so the middle term in the sum expressing $OPT(\mathcal{C})$ is 0. We omit giving complete details. 

The case $n_1>n_2$ is dealt with similarly. 
\end{proof} 

\begin{lemma} Let $\mathcal{C}_1$, $\mathcal{C}_{2}$ be two configurations with $n$ agents s.t. all agents in $\mathcal{C}_1$ are of type 1 and the  transformed configurations of $\mathcal{C}$, $\mathcal{C}_1$ are identical. Then for any  $P\in \{0,1\}$ $f(\mathcal{C},P)=f(\mathcal{C}_1,P)$, while $OPT(\mathcal{C}) \leq OPT(\mathcal{C}_1)$. Hence $app(\mathcal{C},f(\mathcal{C},P))\leq app(\mathcal{C}_1,f(\mathcal{C}_1,P))$. 
\label{lemma:two}
\end{lemma} 
\noindent\textbf{Proof. }The first relation follows directly from the definition of Algorithm~\ref{alephzero-alg}. As for the second, it is enough to prove that $SW(\mathcal{C},Y)\leq SW(\mathcal{C}_1,Y)$ for every $Y\in [0,1]$.  If $1\leq i\leq n$ and agent $i$ is of type $1$ in $\mathcal{C}$ (and in $\mathcal{C}_1$) then $u_i(\mathcal{C},Y)=u_{i}(\mathcal{C}_1,Y)$. Otherwise, if $i$ is of type 0 and location $x_i$ in $\mathcal{C}$ (but, of course, of type $1$ in $\mathcal{C}_1$) then its location in $\mathcal{C}_1$ is $x_{i}^{*}$ and $u_{i}(\mathcal{C},Y)=d(x_i,Y)$, while $u_{i}(\mathcal{C}_1,Y)=1-d(x_{i}^{*},Y)$. Since $d(x_i,Y)+d(x_i^{*},Y)\leq 1$ (\cite{zou2015facility}, Lemma 4) Lemma~\ref{lemma:two} follows. 

\noindent\textbf{Consistency: } We will use again Lemma~\ref{lemma:two}. To prove that the algorithm is $\frac{3-\lambda}{1+\lambda}$ consistent we only need to show that for all configurations $\mathcal{C}$ where all agents are of type 1, if $1=P_{C}=y^{*}$ then $
\frac{OPT(\mathcal{C})}{F(\mathcal{C},P_{C})}\leq \frac{3-\lambda}{1+\lambda}$. 
If the algorithm outputs point $P_{C}$ then the approximation ratio of the algorithm is 1. The remaining  case is when the algorithm outputs the value $0$.  This means that $n_1>n_2$. Furthermore, since $y^{*}=1$, we must have $x_i=1$ for all $i\geq \lceil n/2 \rceil$. 

Let $n=2k$. Just as in the case of proving robustness $x_k\leq \frac{1-\lambda}{2}$. We have 
$OPT(\mathcal{C})=(n-k)\cdot 1 +\sum_{i=1}^{k} x_i.$ On the other hand $F(\mathcal{C},P_{C})=(n-k)\cdot 1 + \sum_{i=1}^{k} 1-x_i=n+\sum_{i=1}^{k} x_i$. So $
\frac{OPT(\mathcal{C})}{F(\mathcal{C},P_{C})}\leq \frac{n-k+k\frac{1-\lambda}{2}}{n-k\frac{1-\lambda}{2}}= \frac{3-\lambda}{1+\lambda}$. 

If $n=2k+1$ then $OPT(\mathcal{C})=n-k-1 +\sum_{i=1}^{k} x_i,$ while $F(\mathcal{C},P_{C})=(n-k)\cdot 1 + \sum_{i=1}^{k} 1-x_i=n-\sum_{i=1}^{k} x_i$. So 
$
\frac{OPT(\mathcal{C})}{F(\mathcal{C},P_{C})}< \frac{k+k\frac{1-\lambda}{2}}{2k-k\frac{1-\lambda}{2}}= \frac{3-\lambda}{1+\lambda}$.

\section{Related Work} 
For lack of space we only give the briefest sketch of related work. 

As mentioned in the introduction, our work is naturally related to the emerging literature on algorithms with predictions, \cite{algo-predictions,purohit2018improving, algo-predictions-webpage}. In particular, the robustness/consistency model was made popular (in the context of the competitive analysis of  ski-rental) by \cite{purohit2018improving}, and further studied in many papers. 
The literature on mechanism design with predictions is, on the other hand, very recent and limited, 
only including, to our best knowledge, references \cite{xu2022mechanism,agrawal2022learning,gkatzelis2022improved,balkanski2022strategyproof}.

The facility location has a long history of investigation, both in a classical (optimization) sense \cite{morris1988facilities}, as well as in a strategic setting (e.g. \cite{procaccia2013approximate,meir2019strategyproof,aziz2020facility}). The classical approaches to obnoxious version of facility location have an equally long history, starting perhaps with \cite{zelinka1968medians}, as a topic in graph theory/operations research, and  including, e.g., \cite{church1978locating,ting1984linear,tamir1991obnoxious}. The strategic approach started in \cite{ye2015strategy} (journal version in \cite{mei2018mechanism}). Other directly relevant related references are \cite{han2012moneyless, ibara2012characterizing} that characterize stategyproof mechanisms in two overlapping but slightly settings, that both generalize obnoxious facility location, as well as \cite{oomine2016characterizing}, that characterize possible outputs of group strategyproof mechanisms. Further variations include, e.g., \cite{cheng2013obnoxious,feigenbaum2015strategyproof,oomine2017,fukui2019lambda,li2021obnoxious}. An interesting, alternative strategic approach based on voting, that studies a so-called \emph{anti-Condorcet point}, and would deserve further attention is given in \cite{labbe1990location}.

\label{sec:related} 
\section{Conclusions and Open Problems} 
\label{sec:conclusions} 

Our main contribution is showing that obnoxious facility location can be analyzed in the framework of mechanism design with predictions. A number of issues raised by our work remain open. The list of the most interesting ones includes, we believe, the following:
\begin{enumerate}[leftmargin=*]
\item In the 2-d case can we give a \textbf{group-strategyproof mechanism} displaying the optimum tradeoff ? 
\item What are the optimal tradeoffs between robustness and consistency in the case of $d$-dimensional hypercubes, $d\geq 3$ ? 
\item Can we give an algorithm for trees that overcomes the issue(s) discussed in Remark~\ref{rem1} ? 
\item What about randomized algorithms ? Such bounds exist in the case without predictions. We  attempted to give such bounds but failed, due to the complexity of resulting calculations. 
\end{enumerate} 

\bibliographystyle{unsrt}
\bibliography{/Users/gistrate/Dropbox/texmf/bibtex/bib/bibtheory}

\end{document}